\documentclass[11pt]{article}

\usepackage[left=2.5cm,right=2.5cm,top=2.5cm,bottom=3cm]{geometry}
\usepackage[affil-it]{authblk}
\usepackage{amsmath}
\usepackage{amssymb}
\usepackage{amsthm}
\usepackage{graphicx}
\usepackage{subcaption}

\usepackage{float}
\usepackage{tikz}
\usetikzlibrary{decorations.pathmorphing}
\usetikzlibrary{calc}

\theoremstyle{plain}
\newtheorem{theorem}[equation]{Theorem}
\newtheorem{corollary}[equation]{Corollary}
\theoremstyle{remark}
\newtheorem{observation}{Observation}
\newtheorem{case}{Case}

\title{Proper circular arc graphs as intersection graphs of paths on a grid}
\author[1]{Esther Galby}
\author[2]{Mar\'\i a P\'\i a Mazzoleni}
\author[1]{Bernard Ries}
\affil[1]{\small University of Fribourg, Department of Informatics, Decision Support $\&$ Operations Research,
Fribourg, Switzerland, esther.galby@unifr.ch, bernard.ries@unifr.ch}
\affil[2]{\small Universidad Nacional de La Plata. Facultad de Ciencias Exactas.
Departamento de Matem\'atica. La Plata, Argentina. / CONICET. pia@mate.unlp.edu.ar}

\date{}

\begin{document}
\maketitle
\setcounter{footnote}{0}
\parindent=0cm

\begin{abstract}
In this paper we present a characterisation, by an infinite family of minimal forbidden induced subgraphs, of proper circular arc graphs which are intersection graphs of paths on a grid, where each path has at most one bend (turn).
\end{abstract}



\section{Introduction}

Intersection graphs of geometrical objects in the plane are among the most studied graph classes and have applications in various domains such as for instance biology, statistics, psychology and computing (see \cite{McMc}). We define the \textit{intersection graph} $G$ of a family $\mathcal{F}$ of non empty sets as the graph whose vertices correspond to the elements of $\mathcal{F}$, and two vertices are adjacent in $G$ if and only if the corresponding elements in $\mathcal{F}$ have a non-empty intersection. 

Golumbic et al. introduced in \cite{Golumbic} the class of \textit{edge intersection graphs of paths on a grid} (\textit{EPG graphs}), i.e. graphs for which there exists a collection of nontrivial paths on a rectangular grid in one-to-one correspondance with their vertex set, such that two vertices are adjacent if and only if the corresponding paths share at least one edge of the grid, and showed that every graph is in fact an EPG graph. A natural restriction which was thereupon considered, suggests to limit the number of \textit{bends} (i.e. 90 degrees turns at a grid-point) that a path may have; for $k \geq 0$, the class \textit{$B_k$-EPG} consists of those EPG graphs admitting a representation in which each path has at most $k$ bends. 

Since their introduction, $B_k$-EPG graphs have been extensively studied from several points of view (see for instance \cite{NCA,Ries,biedl,cohen1,francis,Golumbic,heldt1,heldt2,pergel,BRies}). One major interest is the so-called \textit{bend number}; for a graph class $\mathcal{G}$, the \textit{bend number} of $\mathcal{G}$ is the minimum integer $k\geq 0$ such that every graph $G\in \mathcal{G}$ is a $B_k$-EPG graph. The problem of determining the bend number of graph classes has been widely investigated (see for instance \cite{biedl,francis,Golumbic,heldt1} for planar graphs, Halin graphs, line graphs, outerplanar graphs). 

Since $B_0$-EPG graphs are equivalent to the well-studied class of interval graphs, a particular attention has been paid to $B_1$-EPG graphs. The authors in \cite{heldt2} showed that recognising $B_1$-EPG graphs is an NP-complete problem, a result which was further extended to $B_2$-EPG graphs in \cite{pergel}. Therefore, special graph classes were considered. For instance, the authors in \cite{Ries} provided characterisations of some subclasses of chordal graphs which are $B_1$-EPG by families of minimal forbidden induced subgraphs; in \cite{cohen1}, the authors presented a characterisation of cographs that are $B_1$-EPG and provided a linear time recognition algorithm.

In this paper, we are interested in a subclass of circular arc graphs (CA for short), namely \textit{proper circular arc graphs}. In \cite{NCA}, the authors showed that CA graphs are $B_3$-EPG and further proved that normal circular arc graphs have bend number equal to 2, a result from which we can easily deduce that the bend number of proper circular arc graphs is 2 (see Section \ref{sec:prelim}). They also considered additional constraints on the EPG representations by demanding that the union of the paths lies on the boundary of a rectangle of the grid (\textit{EPR graphs}). Similarly to EPG graphs, they defined for $k \geq 0$ the class \textit{$B_k$-EPR} and proved that not all circular arc graphs are $B_3$-EPR (it is easily seen that CA = $B_4$-EPR = EPR). With the intent of pursuing the work done in \cite{NCA}, we here provide a characterisation of proper circular arc graphs that are $B_1$-EPG by a family of minimal forbidden induced subgraphs (see Section \ref{sec:proper}) which is a first step towards characterising the minimal graphs in (CA $\cap$ $B_2$-EPG) $\backslash$ (CA $\cap$ $B_1$-EPG). We conclude Section \ref{sec:proper} by noting that a characterisation by a family of minimal forbidden induced subgraphs of proper circular arc graphs which are $B_1$-EPR easily follows from \cite{NCA} and \cite{PHCA}. 

\section{Preliminaries}
\label{sec:prelim}

Throughout this paper, all considered graphs are connected, finite and simple. For all graph theoretical terms and notations not defined here, we refer the reader to \cite{Bondy}.

Let $G=(V,E)$ be an undirected graph with vertex set $V$ and edge set $E$. A \textit{clique} (resp. \textit{independent set}) is a subset of vertices that are pairwise adjacent (resp. nonadjacent). If $X_1$ and $X_2$ are two disjoint subsets of vertices, we say that $X_1$ \textit{is complete to} (resp. \textit{is anti-complete to}) $X_2$, which we denote by $X_1 - X_2$ (resp. $X_1 \cdots X_2$), if every vertex in $X_1$ is adjacent (resp. nonadjacent) to every vertex in $X_2$. A \textit{dominating set} $D$ in $G$ is a subset of vertices such that every vertex not in $D$ is adjacent to at least one vertex in $D$.

We denote by $C_n$, $n\geq 3$, the \textit{chordless cycle} on $n$ vertices and by $K_n$, $n \geq 1$, the \textit{complete graph} on $n$ vertices. A \textit{k-wheel}, $k \geq 3$, denoted by $W_k$, is a chordless cycle on $k$ vertices with an additional vertex, referred to as the \textit{center} of the wheel, adjacent to every vertex of the cycle. The \textit{3-sun}, denoted by $S_3$, consists of an independent set $S=\{s_0,s_1,s_2\}$ and a clique $K=\{k_0,k_1,k_2\}$ such that $s_i$ is adjacent to $k_i$ and $k_{i+1}$, $i=0,1,2$, where indices are taken modulo 3. Given a graph $G$ and an integer $k \geq 0$, the \textit{power graph} $G^k$ of $G$ has the same vertex set as $G$ with two vertices being adjacent in $G^k$ if and only if their distance (i.e. the length of a shortest path between the two vertices) in $G$ is at most $k$. 

If $G=(V,E)$ is a graph and $X \subseteq V$ is a subset of vertices, we denote by $G\backslash X$ the graph obtained from $G$ by deleting all vertices in $X$. Equivalently, $G\backslash X$ is the \textit{subgraph of $G$ induced by $V\backslash X$}, denoted by $G[V\backslash X]$. If $X$ consists of a single vertex, say $X=\{x\}$, we simply write $G\backslash x$. The \textit{complement graph} of $G$ is the graph $\overline{G}$ having the same vertex set as $G$ with two vertices being adjacent in $\overline{G}$ if and only if they are nonadjacent in $G$. The \emph{disjoint union of $G_1$ and $G_2$} is denoted by $G_1 \cup G_2$. 

Let $\mathcal H$ be a collection of graphs. For $H\in\mathcal H$, we say that $G$ \emph{contains no induced $H$} if $G$ contains no induced subgraph isomorphic to $H$. A graph is \emph{$\mathcal H$-free} if it contains no induced subgraph isomorphic to some graph belonging to $\mathcal H$.

Recall that an \textit{interval graph} is an intersection graph of intervals on the real line. A graph is said to be \textit{chordal} if it does not contain any chordless cycle of length at least four as an induced subgraph. An independent set of three vertices such that each pair is joined by a path that avoids the neighborhood of the third is called an \textit{asteroidal triple}. The following is a well-known characterisation of interval graphs.

\begin{theorem}[\cite{Lekker}]
\label{theo:interval}
A graph is an interval graph if and only if it is chordal and contains no asteroidal triple.
\end{theorem}

A \textit{circular arc graph} (\textit{CA graph}) is an intersection graph of open arcs on a circle, i.e. a graph $G=(V,E)$ is a circular arc graph if one can associate an open arc on a circle with each vertex such that two vertices are adjacent if and only if their corresponding arcs intersect. If $\mathcal{C}$ denotes the corresponding circle and $\mathcal{A}$ the corresponding set of arcs, then $\mathcal{R} = (\mathcal{C}, \mathcal{A})$ is called a \textit{circular arc representation} of $G$. A circular arc graph having a circular arc representation where no two arcs cover the circle is called a \textit{normal circular arc graph} (\textit{NCA graph}). A circular arc graph having a circular arc representation where no arc properly contains another is called a \textit{proper circular arc graph} (\textit{PCA graph}). It is well known that every PCA graph admits a representation which is simultaneously proper and normal (see \cite{Tucker}); in particular, every PCA graph is a NCA graph. The following theorem provides a minimal forbidden induced subgraph characterisation for PCA graphs (see Fig. \ref{Fig:PCA}).

\begin{theorem}[\cite{PCA}]
\label{PCA}
A graph is a PCA graph if and only if it is $\{G_i, C_{n+4} \cup K_1, \overline{C_{2n + 3} \cup K_1}, \overline{C_{2n+6}}, 1 \leq i \leq 6, n \geq 0\}$-free.
\end{theorem} 

A graph $G$ is a \textit{Helly circular arc graph} (\textit{HCA graph}) if it has a circular arc representation in which any subset of pairwise intersecting arcs has a common point on the circle. A graph that admits a circular arc representation which is simultaneously normal and Helly, i.e. no three arcs or less cover the circle, is called a \textit{normal Helly circular arc graph} (\textit{NHCA graph}). Similarly, one can define the class of \textit{proper Helly circular arc graphs} (\textit{PHCA graphs}) corresponding to those graphs that admit a circular arc representation in which no three arcs cover the circle and no arc properly contains another. It was shown in \cite{sNHCA} that a PCA graph is PHCA if it admits a proper circular arc representation in which no two or three arcs cover the circle; in particular, every PHCA graph is a NHCA graph.

%

\tikzset{
  circ/.style = {circle,draw,fill,inner sep=1pt},
  invisible/.style = {circle,draw=none,inner sep=0pt,font=\tiny}
}

\begin{center}
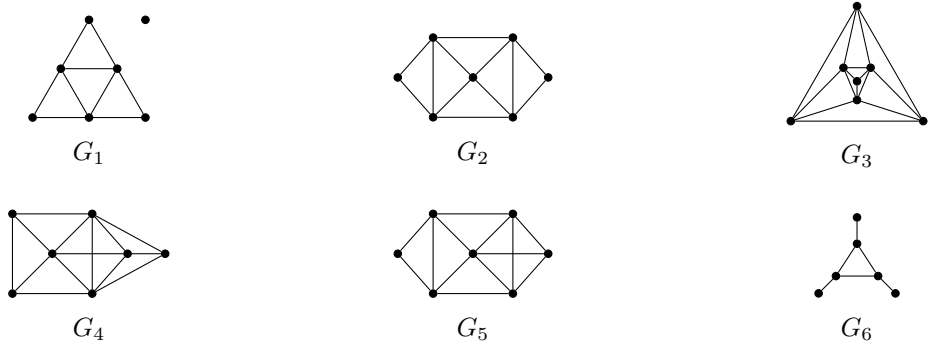
\begin{figure}
\centering
\captionsetup[subfigure]{labelformat=empty}
\begin{minipage}[b]{0.3\textwidth}
\begin{subfigure}[b]{\linewidth}
\centering
\begin{tikzpicture}[node distance=0.75cm]
\node[circ] (c) at (0,0) {};
\node[circ] (b) at (1.5,0) {};
\node[circ] (f) at (0.75,1.3) {};
\node[circ] (g) [right of=f] {};

\draw[-] (c) -- (b) node[circ,midway] (a) {};
\draw[-] (c) -- (f) node[circ,midway] (d) {};
\draw[-] (b) -- (f) node[circ,midway] (e) {};
\draw[-] (a) -- (d)
(d) -- (e)
(e) -- (a);
\end{tikzpicture}
\caption{$G_1$}
\end{subfigure}

\vspace*{5mm}

\begin{subfigure}[b]{\linewidth}
\centering
\begin{tikzpicture}[node distance=0.75cm]
\node[circ] (a) {};
\node[circ] (b) [below left of=a] {};
\node[circ] (c) [below right of=a] {};
\node[circ] (d) [above right of=a] {};
\node[circ] (e) [above left of=a] {};
\node[circ] (f) [right of=a,xshift=0.25cm] {};
\node[circ] (g) [right of=f,xshift=-0.25cm] {};

\draw (a) edge[-] (b)
(a) edge[-] (c)
(a) edge[-] (d)
(a) edge[-] (e)
(a) edge[-] (f)
(b) edge[-] (c)
(b) edge[-] (e)
(c) edge[-] (d)
(c) edge[-] (f)
(c) edge[-] (g)
(d) edge[-] (e)
(d) edge[-] (f)
(d) edge[-] (g)
(f) edge[-] (g);
\end{tikzpicture}
\caption{$G_4$}
\end{subfigure}
\end{minipage}
\begin{minipage}[b]{0.3\textwidth}
\begin{subfigure}[b]{\linewidth}
\centering
\begin{tikzpicture}[node distance=0.75cm]
\node[circ] (a) {};
\node[circ] (b) [below left of=a] {};
\node[circ] (c) [below right of=a] {};
\node[circ] (d) [above right of=a] {};
\node[circ] (e) [above left of=a] {};
\node[circ] (f) [left of=a,xshift=-0.25cm] {};
\node[circ] (g) [right of=a,xshift=0.25cm] {};

\draw (a) edge[-] (b)
(a) edge[-] (c)
(a) edge[-] (d)
(a) edge[-] (e)
(b) edge[-] (c)
(b) edge[-] (e)
(b) edge[-] (f)
(c) edge[-] (d)
(c) edge[-] (g)
(d) edge[-] (e)
(d) edge[-] (g)
(e) edge[-] (f);
\end{tikzpicture}
\caption{$G_2$}
\end{subfigure}

\vspace*{5mm}

\begin{subfigure}[b]{\linewidth}
\centering
\begin{tikzpicture}[node distance=0.75cm]
\node[circ] (a) {};
\node[circ] (b) [below left of=a] {};
\node[circ] (c) [below right of=a] {};
\node[circ] (d) [above right of=a] {};
\node[circ] (e) [above left of=a] {};
\node[circ] (f) [left of=a,xshift=-0.25cm] {};
\node[circ] (g) [right of=a,xshift=0.25cm] {};

\draw (a) edge[-] (b)
(a) edge[-] (c)
(a) edge[-] (d)
(a) edge[-] (e)
(a) edge[-] (g)
(b) edge[-] (c)
(b) edge[-] (e)
(b) edge[-] (f)
(c) edge[-] (d)
(c) edge[-] (g)
(d) edge[-] (e)
(d) edge[-] (g)
(e) edge[-] (f);
\end{tikzpicture}
\caption{$G_5$}
\end{subfigure}
\end{minipage}
\begin{minipage}[b]{0.3\textwidth}
\begin{subfigure}[b]{\linewidth}
\centering
\begin{tikzpicture}[node distance=0.75cm]
\node[circ] (a) {};
\node[circ] (b) [below of=a,yshift=0.5cm] {};
\node[circ] (c) [above right of=a,xshift=-0.35cm,yshift=-0.35cm] {};
\node[circ] (d) [above left of=a,xshift=0.35cm,yshift=-0.35cm] {};
\node[circ] (e) [below left of=b,xshift=-0.35cm,yshift=0.25cm] {};
\node[circ] (f) [below right of=b,xshift=0.35cm,yshift=0.25cm] {};
\node[circ] (g) [above of=a,yshift=0.25cm] {};

\draw (a) edge[-] (b)
(a) edge[-] (c)
(a) edge[-] (d)
(b) edge[-] (c)
(b) edge[-] (d)
(b) edge[-] (e)
(b) edge[-] (f)
(c) edge[-] (d)
(c) edge[-] (f) 
(c) edge[-] (g)
(d) edge[-] (e)
(d) edge[-] (g)
(e) edge[-] (f)
(e) edge[-] (g)
(f) edge[-] (g);
\end{tikzpicture}
\caption{$G_3$}
\end{subfigure}

\vspace*{5mm}

\begin{subfigure}[b]{\linewidth}
\centering
\begin{tikzpicture}[node distance=0.75cm]
\node[invisible] (fake) {};
\node[circ] (a) [below left of=fake,xshift=0.25cm,yshift=0.25cm] {};
\node[circ] (b) [above of=fake,yshift=-0.6cm] {};
\node[circ] (c) [below right of=fake,xshift=-0.25cm,yshift=0.25cm] {};
\node[circ] (d) [below right of=c,xshift=-0.3cm,yshift=0.3cm]  {};
\node[circ] (e) [below left of=a,xshift=0.3cm,yshift=0.3cm] {};
\node[circ] (f) [above of=b,yshift=-0.4cm] {};

\draw (a) edge[-] (b)
(a) edge[-] (c)
(a) edge[-] (e)
(b) edge[-] (c)
(b) edge[-] (f)
(c) edge[-] (d);
\end{tikzpicture}
\caption{$G_6$}
\end{subfigure}
\end{minipage}
\caption{Minimal forbidden induced subgraphs for PCA graphs.}
\label{Fig:PCA}
\end{figure}
\end{center}

Consider a rectangular grid $\mathcal{G}$ where the horizontal lines are referred to as \textit{rows} and the vertical lines as \textit{columns}. A grid-point lying on row $x$ and column $y$ is referred to as $(x,y)$. If $\mathcal{P}$ is a collection of nontrivial simple paths on the grid, the \textit{edge intersection graph $G$ of $\mathcal{P}$} is the graph whose vertex set is in one-to-one correspondance with $\mathcal{P}$ and two vertices are adjacent if and only if the corresponding paths share at least one grid-edge. The path representing some vertex $v$ will be denoted by $P_v$. Then $(\mathcal{G},\mathcal{P})$ is referred to as an \textit{EPG representation} of $G$ or a \textit{k-bend EPG representation} of $G$ if every path of $\mathcal{P}$ has at most $k$-bends (i.e. 90 degrees turns at a grid-point) with $k \geq 0$. The class of graphs admitting a $k$-bend EPG representation is called \textit{$B_k$-EPG}. 

A graph $G$ is said to be an \textit{edge intersection graph of paths on a rectangle} (\textit{EPR graph}) if there exists a set of paths $\mathcal{P}$ on a rectangle $\mathcal{R}$ of the grid in one-to-one correspondance with the vertex set of $G$, where two vertices are adjacent in $G$ if and only if their corresponding paths share at least one grid-edge; $(\mathcal{G},\mathcal{R},\mathcal{P})$ is then referred to as an \textit{EPR representation} of $G$. For $k \geq 0$, we denote by \textit{$B_k$-EPR} the class of graphs for which there exists an EPR representation where every path has at most $k$ bends. 

The authors in \cite{NCA} proved that NCA graphs have a bend number of 2 and presented an infinite family of NCA graphs, namely $\{C_{4k-1}^k, k \geq 2\}$, which are not $B_1$-EPG. Since any $C_{4k-1}^k, k \geq 2$ is in fact a PCA graph, we deduce the following corollary from the fact that PCA $\subset$ NCA.

\begin{corollary}
PCA graphs have a bend number of 2.
\end{corollary}


\section{Proper circular arc $B_1$-EPG graphs}
\label{sec:proper}

As we have seen in Section \ref{sec:prelim}, the bend number of proper circular arc graphs is 2. In this section, we provide a characterisation, by a family of minimal forbidden induced subgraphs (see Fig. \ref{PCAB1}), for PCA graphs which are $B_1$-EPG.

\tikzset{
  circ/.style = {circle,draw,fill,inner sep=1pt},
  nonedge/.style={decorate,decoration={snake,amplitude=.3mm,segment length=1mm},draw},
  clique/.style = {circle,draw,inner sep=1pt,font=\tiny},
  invisible/.style = {circle,draw=none,inner sep=0pt,font=\tiny}
}

\begin{center}
\begin{figure}[h]
\captionsetup[subfigure]{labelformat=empty}
\begin{minipage}[b]{0.24\textwidth}
\begin{subfigure}[b]{\linewidth}
\centering
\begin{tikzpicture}[node distance=0.75cm]
\node[circ] (a) {};
\node[circ] (b) [below left of=a] {};
\node[circ] (c) [below right of=a] {};
\node[circ] (d) [above right of=a] {};
\node[circ] (e) [above left of=a] {};
\node[circ] (f) [above right of=d] {};
\node[circle,draw=none,inner sep=1pt] (fake) [below left of=b] {};

\draw (a) edge[-] (b)
(a) edge[-] (c)
(a) edge[-] (d)
(a) edge[-] (e)
(b) edge[-] (c)
(b) edge[-] (e)
(c) edge[-] (d)
(c) edge[-,bend right=10] (f)
(d) edge[-] (f)
(d) edge[-] (e)
(e) edge[-,bend left=10] (f);
\end{tikzpicture}
\caption{$H_1$}
\end{subfigure}

\vspace*{5mm}

\begin{subfigure}[b]{\linewidth}
\centering
\begin{tikzpicture}[node distance=0.75cm]
\node[circ] (a) {};
\node[circ] (b) [below left of=a] {};
\node[circ] (c) [below right of=a] {};
\node[circ] (d) [above right of=a] {};
\node[circ] (e) [above left of=a] {};
\node[circ] (f) [above right of=d] {};
\node[circ] (g) [above left of=b] {};

\draw (a) edge[-] (b)
(a) edge[-] (c)
(a) edge[-] (d)
(a) edge[-] (e)
(a) edge[-,bend right=20] (f)
(a) edge[nonedge] (g)
(b) edge[-] (c)
(b) edge[-] (e)
(b) edge[-] (g)
(c) edge[-] (d)
(c) edge[-,bend right=10] (f)
(d) edge[-] (f)
(d) edge[-] (e)
(e) edge[-,bend left=10] (f)
(e) edge[-] (g)
(f) edge[-] (g);
\end{tikzpicture}
\caption{$H_5$}
\end{subfigure}
\end{minipage}
\begin{minipage}[b]{0.24\textwidth}
\begin{subfigure}[b]{\linewidth}
\centering
\begin{tikzpicture}[node distance=0.75cm]
\node[circ] (a) {};
\node[circ] (b) [below left of=a] {};
\node[circ] (c) [below right of=a] {};
\node[circ] (d) [above right of=a] {};
\node[circ] (e) [above left of=a] {};
\node[circ] (f) [above right of=d] {};
\node[circ] (g) [below right of=c] {};

\draw (a) edge[-] (b)
(a) edge[-] (c)
(a) edge[-] (d)
(a) edge[-] (e)
(a) edge[-,bend left=20] (f)
(a) edge[-,bend right=20] (g)
(b) edge[-] (c)
(b) edge[-] (e)
(b) edge[-,bend right=10] (g)
(c) edge[-] (d)
(c) edge[-,bend right=10] (f)
(c) edge[-] (g)
(d) edge[-] (f)
(d) edge[-] (e)
(d) edge[-,bend left=10] (g)
(e) edge[-,bend left=10] (f);
\end{tikzpicture}
\caption{$H_2$}
\end{subfigure}

\vspace*{5mm}

\begin{subfigure}[b]{\linewidth}
\centering
\begin{tikzpicture}[node distance=0.75cm]
\begin{pgfinterruptboundingbox}
\node[circle,draw=none] (a) {};
\end{pgfinterruptboundingbox}
\node[circ] (a') [left of=a,xshift=0.5cm] {};
\node[circ] (a'') [right of=a,xshift=-0.5cm] {};
\node[circ] (b) [below left of=a] {};
\node[circ] (c) [below right of=a] {};
\node[circ] (d) [above right of=a] {};
\node[circ] (e) [above left of=a] {};
\node[circ] (f) [above left of=d] {};
\node[circ] (g) [above left of=b] {};

\draw (a') edge[-] (b)
(a') edge[-] (c)
(a') edge[-] (d)
(a') edge[-] (e)
(a') edge[-] (g)
(a') edge[-] (a'')
(a'') edge[-] (b)
(a'') edge[-] (c)
(a'') edge[-] (d)
(a'') edge[-] (e)
(a'') edge[-] (f)
(b) edge[-] (c)
(b) edge[-] (e)
(b) edge[-] (g)
(c) edge[-] (d)
(d) edge[-] (f)
(d) edge[-] (e)
(e) edge[-] (f)
(e) edge[-] (g);
\end{tikzpicture}
\caption{$H_6$}
\end{subfigure}
\end{minipage}
\begin{minipage}[b]{0.24\textwidth}
\begin{subfigure}[b]{\linewidth}
\centering
\begin{tikzpicture}[node distance=0.75cm]
\node[circ] (a) {};
\node[circ] (a') [left of=a,xshift=0.5cm] {};
\node[circ] (a'') [right of=a, xshift=-0.5cm] {};
\node[circ] (b) [below left of=a] {};
\node[circ] (c) [below right of=a] {};
\node[circ] (d) [above right of=a] {};
\node[circ] (e) [above left of=a] {};
\node[circle,draw=none,inner sep=1pt] (fake) [below right of=b]{};

\draw (a) edge[-] (b)
(a) edge[-] (c)
(a) edge[-] (d)
(a) edge[-] (e)
(a) edge[-] (a')
(a) edge[-] (a'')
(a') edge[-] (b)
(a') edge[-] (c)
(a') edge[-] (d)
(a') edge[-] (e)
(a'') edge[-] (b)
(a'') edge[-] (c)
(a'') edge[-] (d)
(a'') edge[-] (e)
(b) edge[-] (c)
(b) edge[-] (e)
(c) edge[-] (d)
(d) edge[-] (e);
\end{tikzpicture}
\caption{$H_3$}
\end{subfigure}

\vspace*{5mm}

\begin{subfigure}[b]{\linewidth}
\centering
\begin{tikzpicture}[node distance=0.75cm]
\node[circ] (a) {};
\node[circ] (b) [below left of=a] {};
\node[circ] (c) [below right of=a] {};
\node[circ] (d) [above right of=a] {};
\node[circ] (e) [above left of=a] {};
\node[circ] (f) [above right of=d,yshift=-0.25cm,xshift=-0.25cm] {};
\node[circ] (g) [below left of=b,yshift=0.25cm,xshift=0.25cm] {};

\draw (a) edge[-] (b)
(a) edge[-] (c)
(a) edge[-] (d)
(a) edge[-] (e)
(a) edge[-,bend left=20] (f)
(a) edge[-,bend left=20] (g)
(b) edge[-] (c)
(b) edge[-] (e)
(b) edge[-] (g)
(c) edge[-] (d)
(c) edge[-,bend right=10] (f)
(c) edge[-,bend left=10] (g)
(d) edge[-] (f)
(d) edge[-] (e)
(e) edge[-,bend left=10] (f)
(e) edge[-,bend right=10] (g)
(f) edge[-,bend right] (g);
\end{tikzpicture}
\caption{$H_7$}
\end{subfigure}
\end{minipage}
\begin{minipage}[b]{0.24\textwidth}
\begin{subfigure}[b]{\linewidth}
\centering
\begin{tikzpicture}[node distance=0.75cm]
\begin{pgfinterruptboundingbox}
\node[circle,draw=none] (a) {};
\end{pgfinterruptboundingbox}
\node[circ] (a') [above of=a,yshift=-0.5cm] {};
\node[circ] (a'') [below of=a, yshift=0.5cm] {};
\node[circ] (b) [below left of=a] {};
\node[circ] (c) [below right of=a] {};
\node[circ] (d) [above right of=a] {};
\node[circ] (e) [above left of=a] {};
\node[circ] (f) [below right of=d] {};
\node[circle,draw=none,inner sep=1pt] (fake) [below left of=b]{};

\draw (a') edge[-] (b)
(a') edge[-] (c)
(a') edge[-] (d)
(a') edge[-] (e)
(a') edge[-] (f)
(a'') edge[-] (b)
(a'') edge[-] (c)
(a'') edge[-] (d)
(a'') edge[-] (e)
(a'') edge[-] (f)
(b) edge[-] (c)
(b) edge[-] (e)
(c) edge[-] (d)
(c) edge[-] (f)
(d) edge[-] (f)
(d) edge[-] (e);
\end{tikzpicture}
\caption{$H_4$}
\end{subfigure}

\vspace*{5mm}

\begin{subfigure}[b]{\linewidth}
\centering
\begin{tikzpicture}[node distance=0.75cm]
\node[circ] (a) {};
\node[circ] (b) [below left of=a] {};
\node[circ] (c) [below right of=a] {};
\node[circ] (d) [above right of=a] {};
\node[circ] (e) [above left of=a] {};
\node[circ] (f) [above left of=d] {};
\node[circ] (g) [below right of=d] {};

\draw (a) edge[-] (b)
(a) edge[-] (c)
(a) edge[-] (d)
(a) edge[-] (e)
(a) edge[-] (g)
(b) edge[-] (c)
(b) edge[-] (e)
(c) edge[-] (d)
(c) edge[-] (g)
(d) edge[-] (f)
(d) edge[-] (e)
(d) edge[-] (g)
(e) edge[-] (f)
(f) edge[-,bend left] (g);
\end{tikzpicture}
\caption{$H_8$}
\end{subfigure}
\end{minipage}
\caption{PCA graphs which are minimally non $B_1$-EPG (the serpentine line connecting two vertices indicates the existence of either an edge or a nonedge between those two vertices).}
\label{PCAB1}
\end{figure}
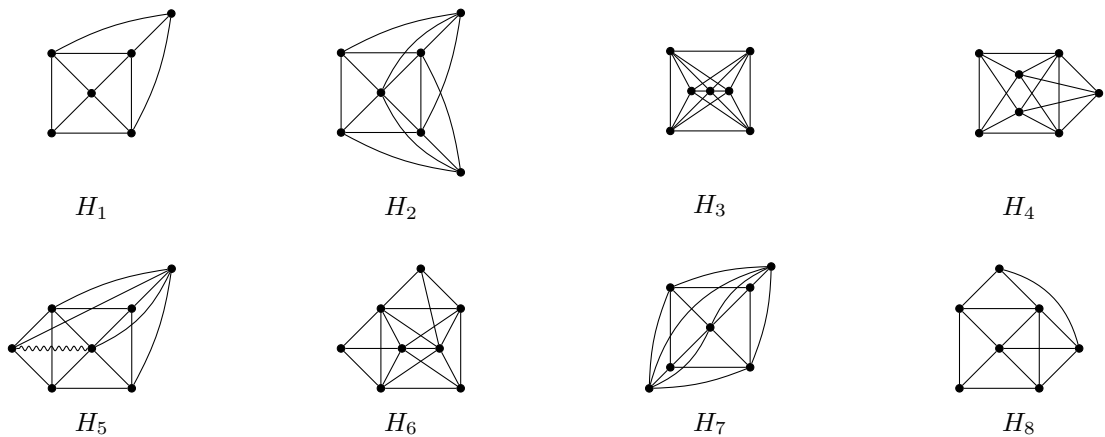
\end{center}

\begin{theorem}
\label{thmPCA}
Let $G$ be a PCA graph. Then $G$ is $B_1$-EPG if and only if $G$ is $\{H_i, C_{4k-1}^k, 1 \leq i \leq 8, k \geq 2\}$-free.
\end{theorem}

\begin{proof}

\textit{\underline{Necessary condition}.} Let us show that for all $1 \leq i \leq 8$, $H_i$ is not $B_1$-EPG. Observe that all these graphs contain an induced 4-wheel; we denote by $a$, $b$, $c$ and $d$ the four vertices of the 4-cycle of the 4-wheel, and $e$ (one of) its center(s). As shown in \cite{Ries}, this 4-cycle can only be represented by either a true pie or a false pie as $\mathcal{P}_e$ should intersect all four corresponding paths of the cycle (see Fig. \ref{b1}). 

Assume henceforth that the 4-wheel is represented by a true pie using column $x$ and row $y$ of the grid (a similar reasoning applies if it is represented by a false pie). Then, $\mathcal{P}_e$ lies either on column $x$ or row $y$ and strictly contains the grid-point $(x,y)$, since it must intersect every path of the 4-cycle (see Fig. \ref{tworep}). Consequently, $H_3$ can not be $B_1$-EPG.

If a vertex $v$ is adjacent to three vertices of the 4-cycle, say $a$, $b$ and $c$, its associated path must contain row $x$ and/or column $y$, since it has to intersect $\mathcal{P}_a$, $\mathcal{P}_b$ and $\mathcal{P}_c$. However, $\mathcal{P}_v$ can not lie entirely on column $x$ or row $y$ as it would otherwise intersect $\mathcal{P}_d$. Hence, it lies on both similarly to $\mathcal{P}_b$; but then $\mathcal{P}_v$ necessarily intersects $\mathcal{P}_e$, which implies that $H_1$ is not $B_1$-EPG (see Fig. \ref{three}). 

If a second vertex $w$ intersects three different vertices of the 4-cycle, we distinguish two cases. Either those vertices are $b$, $c$ and $d$ (the case where they are $d$, $a$ and $b$ is symmetric), then as shown previously, $\mathcal{P}_w$ lies on column $x$ and row $y$ similarly to $\mathcal{P}_c$, and therefore intersects $\mathcal{P}_v$. Or they are $c$, $d$ and $a$ in which case $\mathcal{P}_w$ lies on column $x$ and row $y$ similarly to $\mathcal{P}_d$, and consequently shares only the grid-point $(x,y)$ with $\mathcal{P}_v$. Hence, neither $H_2$ nor $H_7$ are $B_1$-EPG. 

If $z$ is a vertex adjacent to exactly two consecutive vertices of the 4-cycle, say $a$ and $d$, then $\mathcal{P}_z$ uses row $y$, where $\mathcal{P}_a$ and $\mathcal{P}_d$ intersect, without strictly containing the grid-point $(x,y)$ as otherwise it would intersect other vertices of the 4-cycle (note that $\mathcal{P}_z$ may then only intersect centers of the 4-wheel which are adjacent and lie on row $y$). Hence, $H_4$ and $H_5$ are not $B_1$-EPG (see Fig. \ref{two}). 

Finally, since a vertex $t$ adjacent to only $a$ and $b$ would use column $x$, where $\mathcal{P}_a$ and $\mathcal{P}_b$ intersect, without strictly containing the grid-point $(x,y)$, $H_6$ and $H_8$ cannot be $B_1$-EPG. We conclude this part of the proof by noticing that $C_{4k-1}^k \not \in B_1$-EPG for $k \geq 2$, as shown in \cite{NCA}.

\begin{center}
\begin{figure}[h]
\begin{minipage}[b]{0.5\textwidth}
\centering
\begin{subfigure}[b]{\linewidth}
\begin{minipage}[b]{.45\textwidth}
\centering
\begin{tikzpicture}
\coordinate (o) at (-0.75,0);
\coordinate (e) at (0.75,0);
\coordinate (n) at (0,0.75);
\coordinate (s) at (0,-0.75);

\draw[thick] ($(o) + (0,0.05)$) -- (-0.05,0.05) -- ($(n) + (-0.05,0)$);
\draw[thick] ($(o) + (0,-0.05)$) -- (-0.05,-0.05) -- ($(s) + (-0.05,0)$);
\draw[thick] ($(e) + (0,0.05)$) -- (0.05,0.05) -- ($(n) + (0.05,0)$);
\draw[thick] ($(e) + (0,-0.05)$) -- (0.05,-0.05) -- ($(s) + (0.05,0)$);
\end{tikzpicture}
\end{minipage}
\begin{minipage}[b]{0.45\textwidth}
\centering
\begin{tikzpicture}
\coordinate (o) at (-0.75,0);
\coordinate (e) at (0.75,0);
\coordinate (n) at (0,0.75);
\coordinate (s) at (0,-0.75);

\draw[thick] ($(o) + (0,0.1)$) -- (-0.1,0.1) -- ($(n) + (-0.1,0)$);
\draw[thick] (o) -- (e);
\draw[thick] (n) -- (s);
\draw[thick] ($(e) + (0,-0.1)$) -- (0.1,-0.1) -- ($(s) + (0.1,0)$);
\end{tikzpicture}
\end{minipage}
\caption{A true pie (left) and  a false pie (right).}
\label{b1}
\end{subfigure}

\vspace*{5mm}

\begin{subfigure}[b]{\linewidth}
\centering
\begin{minipage}[b]{.45\textwidth}
\centering
\begin{tikzpicture}[scale=.8]
\coordinate (o) at (-1,0);
\coordinate (e) at (1,0);
\coordinate (n) at (0,1);
\coordinate (s) at (0,-1);

\draw[thick] (o) -- ($(e)+ (0.2,0)$) node[pos=0.99,above] {\tiny $\mathcal{P}_e$};
\draw[thick] ($(o) + (0,0.1)$) -- (-0.05,0.1) -- ($(n) + (-0.05,0)$) node[pos=0.99,left] {\tiny $\mathcal{P}_a$};
\draw[thick] ($(o) + (0,-0.1)$) -- (-0.05,-0.1) -- ($(s) + (-0.05,0)$) node[pos=0.99,left] {\tiny $\mathcal{P}_d$};
\draw[thick] ($(e) + (0,0.1)$) -- (0.05,0.1) -- ($(n) + (0.05,0)$) node[pos=0.99,right] {\tiny $\mathcal{P}_b$};
\draw[thick] ($(e) + (0,-0.1)$) -- (0.05,-0.1) -- ($(s) + (0.05,0)$) node[pos=0.99,right] {\tiny $\mathcal{P}_c$};
\draw[thick] ($(o) + (0,0.2)$) -- ($(o) + (0.6,0.2)$) node[midway,label=above:{\tiny $\mathcal{P}_z$}] {};
\end{tikzpicture}
\end{minipage}
\begin{minipage}[b]{.45\textwidth}
\centering
\begin{tikzpicture}[scale=.8]
\coordinate (o) at (-1,0);
\coordinate (e) at (1,0);
\coordinate (n) at (0,1);
\coordinate (s) at (0,-1);

\draw[thick] ($(o) + (.55,0)$) -- ($(e)+ (0.2,0)$) node[pos=0.99,above] {\tiny $\mathcal{P}_e$};
\draw[thick] ($(o) + (0,0.1)$) -- (-0.05,0.1) -- ($(n) + (-0.05,0)$) node[pos=0.99,left] {\tiny $\mathcal{P}_a$};
\draw[thick] ($(o) + (0,-0.1)$) -- (-0.05,-0.1) -- ($(s) + (-0.05,0)$) node[pos=0.99,left] {\tiny $\mathcal{P}_d$};
\draw[thick] ($(e) + (0,0.1)$) -- (0.05,0.1) -- ($(n) + (0.05,0)$) node[pos=0.99,right] {\tiny $\mathcal{P}_b$};
\draw[thick] ($(e) + (0,-0.1)$) -- (0.05,-0.1) -- ($(s) + (0.05,0)$) node[pos=0.99,right] {\tiny $\mathcal{P}_c$};
\draw[thick] ($(o) + (0,0.2)$) -- ($(o) + (0.45,0.2)$) node[midway,label=above:{\tiny $\mathcal{P}_z$}] {};
\end{tikzpicture}
\end{minipage}
\caption{Vertex $z$ is adjacent to $a$, $d$ and a center (left) or not (right).}
\label{two}
\end{subfigure}
\end{minipage}
\begin{minipage}[b]{.5\textwidth}
\centering
\begin{subfigure}[b]{\linewidth}
\begin{minipage}[b]{0.45\textwidth}
\centering
\begin{tikzpicture}[scale=.8]
\coordinate (o) at (-1,0);
\coordinate (e) at (1,0);
\coordinate (n) at (0,1);
\coordinate (s) at (0,-1);

\draw[thick] (o) -- ($(e)+ (0.2,0)$) node[pos=0.99,above] {\tiny $\mathcal{P}_e$};
\draw[thick] ($(o) + (0,0.1)$) -- (-0.05,0.1) -- ($(n) + (-0.05,0)$) node[pos=0.99,left] {\tiny $\mathcal{P}_a$};
\draw[thick] ($(o) + (0,-0.1)$) -- (-0.05,-0.1) -- ($(s) + (-0.05,0)$) node[pos=0.99,left] {\tiny $\mathcal{P}_d$};
\draw[thick] ($(e) + (0,0.1)$) -- (0.05,0.1) -- ($(n) + (0.05,0)$) node[pos=0.99,right] {\tiny $\mathcal{P}_b$};
\draw[thick] ($(e) + (0,-0.1)$) -- (0.05,-0.1) -- ($(s) + (0.05,0)$) node[pos=0.99,right] {\tiny $\mathcal{P}_c$};
\end{tikzpicture}
\end{minipage}
\begin{minipage}[b]{0.45\textwidth}
\centering
\begin{tikzpicture}[scale=.8]
\coordinate (o) at (-1,0);
\coordinate (e) at (1,0);
\coordinate (n) at (0,1);
\coordinate (s) at (0,-1);

\draw[thick] (s) -- ($(n)+ (0,0.2)$) node[pos=0.99,above] {\tiny $\mathcal{P}_e$};
\draw[thick] ($(o) + (0,0.05)$) -- (-0.1,0.05) -- ($(n) + (-0.1,0)$) node[pos=0.99,left] {\tiny $\mathcal{P}_a$};
\draw[thick] ($(o) + (0,-0.05)$) -- (-0.1,-0.05) -- ($(s) + (-0.1,0)$) node[pos=0.99,left] {\tiny $\mathcal{P}_d$};
\draw[thick] ($(e) + (0,0.05)$) -- (0.1,0.05) -- ($(n) + (0.1,0)$) node[pos=0.99,right] {\tiny $\mathcal{P}_b$};
\draw[thick] ($(e) + (0,-0.05)$) -- (0.1,-0.05) -- ($(s) + (0.1,0)$) node[pos=0.99,right] {\tiny $\mathcal{P}_c$};
\end{tikzpicture}
\end{minipage}
\caption{Representations of $W_4$ with a true pie.} 
\label{tworep}
\end{subfigure}

\vspace*{5mm}

\begin{subfigure}[b]{\linewidth}
\begin{minipage}[b]{0.45\textwidth}
\centering
\begin{tikzpicture}[scale=.8]
\coordinate (o) at (-1,0);
\coordinate (e) at (1,0);
\coordinate (n) at (0,1);
\coordinate (s) at (0,-1);

\draw[thick] (o) -- ($(e)+ (0.2,0)$) node[pos=0.99,above] {\tiny $\mathcal{P}_e$};
\draw[thick] ($(o) + (0,0.1)$) -- (-0.05,0.1) -- ($(n) + (-0.05,0)$) node[pos=0.99,left] {\tiny $\mathcal{P}_a$};
\draw[thick] ($(o) + (0,-0.1)$) -- (-0.05,-0.1) -- ($(s) + (-0.05,0)$) node[pos=0.99,left] {\tiny $\mathcal{P}_d$};
\draw[thick] ($(e) + (0,0.1)$) -- (0.05,0.1) -- ($(n) + (0.05,0)$) node[pos=0.99,right] {\tiny $\mathcal{P}_b$};
\draw[thick] ($(e) + (0,-0.1)$) -- (0.05,-0.1) -- ($(s) + (0.05,0)$) node[pos=0.99,right] {\tiny $\mathcal{P}_c$};
\draw[thick] ($(n) + (0.15,-0.2)$) -- (0.15, 0.2) -- ($(e) + (-0.2,0.2)$) node[pos=0.99,above] {\tiny $\mathcal{P}_v$};
\end{tikzpicture}
\end{minipage}
\begin{minipage}[b]{0.45\textwidth}
\centering
\begin{tikzpicture}[scale=.8]
\coordinate (o) at (-1,0);
\coordinate (e) at (1,0);
\coordinate (n) at (0,1);
\coordinate (s) at (0,-1);

\draw[thick] (s) -- ($(n)+ (0,0.2)$) node[pos=0.99,above] {\tiny $\mathcal{P}_e$};
\draw[thick] ($(o) + (0,0.05)$) -- (-0.1,0.05) -- ($(n) + (-0.1,0)$) node[pos=0.99,left] {\tiny $\mathcal{P}_a$};
\draw[thick] ($(o) + (0,-0.05)$) -- (-0.1,-0.05) -- ($(s) + (-0.1,0)$) node[pos=0.99,left] {\tiny $\mathcal{P}_d$};
\draw[thick] ($(e) + (0,0.05)$) -- (0.1,0.05) -- ($(n) + (0.1,0)$) node[pos=0.99,right] {\tiny $\mathcal{P}_b$};
\draw[thick] ($(e) + (0,-0.05)$) -- (0.1,-0.05) -- ($(s) + (0.1,0)$) node[pos=0.99,right] {\tiny $\mathcal{P}_c$};
\draw[thick] ($(n) + (0.2,-0.2)$) -- (0.2,0.15) -- ($(e) + (0,0.15)$) node[pos=0.99,above] {\tiny $\mathcal{P}_v$};
\end{tikzpicture}
\end{minipage}
\caption{Vertex $v$ is adjacent to $a$, $b$ and $c$.}
\label{three}
\end{subfigure}
\end{minipage}
\caption{$B_1$-EPG representations.}
\end{figure}
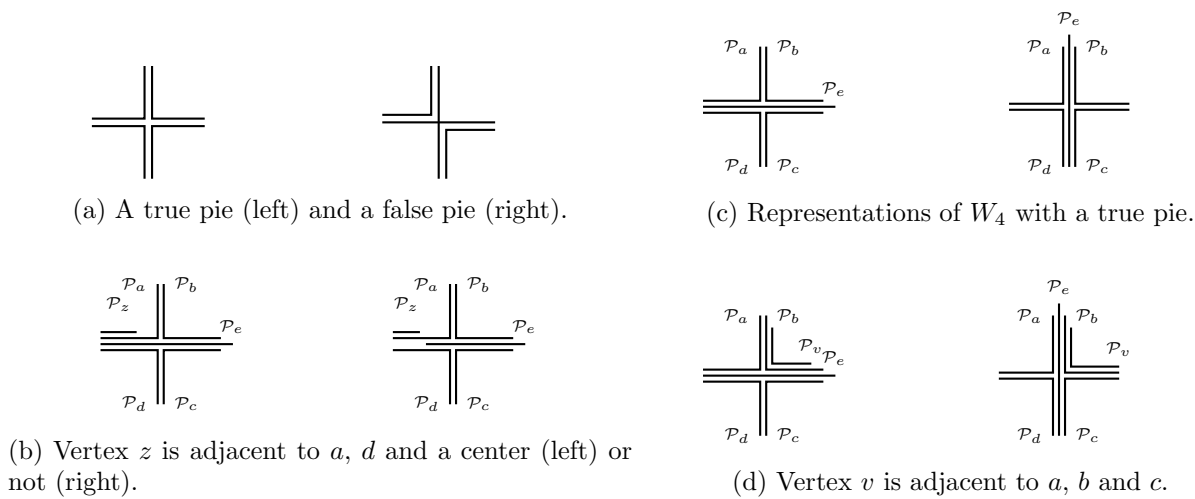
\end{center}

\textit{\underline{Sufficient condition}.} Let $G=(V,E)$ be a PCA graph which is $\{H_i, C_{4k-1}^k, 1 \leq i \leq 8, k \geq2\}$-free. Consider a normal proper representation $\mathcal{R} = (\mathcal{C}, \mathcal{A})$ of $G$, where $\mathcal{C}$ is a circle and $\mathcal{A}$ is a set of open arcs of $\mathcal{C}$ in one-to-one correspondance with the vertices of $G$ (notice that such a representation exists due to \cite{Tucker}). Before turning to the proof, let us first make the following observation.

\begin{observation}
\label{obs1}
Any set of arcs in $\mathcal{A}$ covering the circle $\mathcal{C}$ corresponds to a dominating set in $G$. In particular, if $G$ contains a 4-wheel as an induced subgraph, then $G$ has a dominating triangle.
\end{observation}

\textit{Proof.} It is clear that such a set of arcs corresponds to a dominating set in $G$. If $G$ contains a 4-wheel, then the arcs corresponding to 4-cycle $C$ of the 4-wheel cover the circle $\mathcal{C}$. But then the arc representing the center of the 4-wheel together with two arcs corresponding to two vertices of $C$ must also cover $\mathcal{C}$, i.e. $G$ has a dominating triangle. $\diamond$\\

If we assume that no three arcs in $\mathcal{A}$ cover $\mathcal{C}$, then $G$ is PHCA and the result follows from the fact that PHCA $\cap \{C_{4k -1}^k\}_{k \geq 2}$-free $\subset$ NHCA $\cap \{C_{4k -1}^k\}_{k \geq 2}$-free $=B_1$-EPR (see \cite{NCA}). Hence, we may assume that there exist three arcs in $\mathcal{A}$ covering $\mathcal{C}$.  If $G$ contains a 4-wheel, let $C = \{x_1,x_2,x_3\}$ denote the dominating triangle following from Observation \ref{obs1}, with $x_1$ being the center of the 4-wheel. Otherwise, let $C = \{x_1, x_2, x_3\}$ be any triangle whose corresponding arcs cover $\mathcal{C}$. In both cases, each vertex is adjacent to at least two vertices of $C$. Indeed, if $G$ contains a $4$-wheel, then this follows from the fact that $G$ is claw-free (see Theorem \ref{PCA}); if $G$ contains no $4$-wheel, this follows from the fact that $\mathcal{R}$ is a proper representation. For $j \in \{1,2,3\}$, denote by $\mathcal{A}_{j,j+1} = \{x \in V ~|~ xx_{j-1} \not \in E\}$, where indices are taken modulo 3, the subset of vertices adjacent to only $x_j$ and $x_{j+1}$. Note that each $\mathcal{A}_{j,j+1}$ is a clique as $G$ would otherwise contain an induced claw, namely $x_{j+1}, x_{j-1}, x, x'$ for any two $x, x' \in \mathcal{A}_{j,j+1}$ such that $xx' \not \in E$. Similarly, consider the subset of vertices $\mathcal{A}_c = \{x \in V ~|~ \forall j \in \{1,2,3\}, xx_j \in E\}$ adjacent to all three vertices of $C$. We now distinguish cases depending on whether $G$ contains a 4-wheel as an induced subgraph or not.

\begin{case}[\textit{$G$ contains an induced 4-wheel}] 
According to the above, $\mathcal{A}_{1,2}$ and $\mathcal{A}_{1,3}$ are not anticomplete. Thus, there exist $x_4 \in \mathcal{A}_{1,3}$ and $x_5 \in \mathcal{A}_{1,2}$ such that $x_4x_5 \in E$, which together with $x_2$ and $x_3$ form the 4-cycle $C'$ of the 4-wheel. Since $C'$ is dominating and $G$ has no induced claw, each remaining vertex of $G$ is adjacent to at least two vertices of $C'$. Consider accordingly the subset of vertices $\mathcal{A}_{3,4}$ (resp. $\mathcal{A}_{4,5}$, $\mathcal{A}_{5,2}$) adjacent to only $x_3$ and $x_4$ (resp. $x_4$ and $x_5$, $x_5$ and $x_2$), the subset of vertices $\mathcal{A}_2$ (resp. $\mathcal{A}_3$; $\mathcal{A}_4$; $\mathcal{A}_5$) adjacent to only $x_5$, $x_2$ and $x_3$ (resp. $x_2$, $x_3$ and $x_4$; $x_3$, $x_4$ and $x_5$; $x_4$, $x_5$ and $x_2$) and the subset of vertices $\mathcal{A}_{c'} = \{ x \in V ~|~ \forall i \in \{2,3,4,5\}, xx_i \in E\}$ adjacent to all vertices of $C'$ (note that $x_1 \in \mathcal{A}_{c'}$). Since $G$ contains no induced claw, each $\mathcal{A}_{j,j+1}$ is a clique, as well as each $\mathcal{A}_j$. Furthermore, since $G$ contains no induced:

\begin{itemize}
\item[$\bullet$] $H_1$, each $\mathcal{A}_j$ is complete to $\mathcal{A}_{c'}$, for $j=2,3,4,5$; 
\item[$\bullet$] $H_2$, we have $\mathcal{A}_2 - \mathcal{A}_3 - \mathcal{A}_4 - \mathcal{A}_5 - \mathcal{A}_2$;
\item[$\bullet$] $H_7$, we have $\mathcal{A}_2 \cdots \mathcal{A}_4$ and $\mathcal{A}_3 \cdots \mathcal{A}_5$;
\item[$\bullet$] $H_5$, we have $\mathcal{A}_{j,k} \cdots \mathcal{A}_i$ for all $ i \neq j,k$ with $(j,k) \in \{(2,3), (3,4), (4,5), (5,2)\}$;
\item[$\bullet$] $H_8$ and 5-wheel (by Theorem \ref{PCA}), we have $\mathcal{A}_{2,3} \cdots \mathcal{A}_{3,4} \cdots \mathcal{A}_{4,5} \cdots \mathcal{A}_{5,2} \cdots \mathcal{A}_{2,3}$;
\item[$\bullet$] $\overline{C_6}$ (by Theorem \ref{PCA}), we have $\mathcal{A}_{2,3} \cdots \mathcal{A}_{4,5}$ and $\mathcal{A}_{3,4} \cdots \mathcal{A}_{5,2}$;
\item[$\bullet$] claw (by Theorem \ref{PCA}), we have $\mathcal{A}_j - \mathcal{A}_{j,k} - \mathcal{A}_{k}$ for $(j,k) \in \{(2,3), (3,4), (4,5), (5,2)\}$.
\end{itemize}

Now, if we assume that both $\mathcal{A}_{2,3}$ and $\mathcal{A}_{4,5}$ are nonempty, then both are complete to $\mathcal{A}_{c'}$ as $G$ would otherwise contain either $G_2$ or $G_5$ as an induced subgraph (see Fig. \ref{Fig:PCA}). But then $\mathcal{A}_{c'}$ is a clique since $G$ does not contain $H_4$ as an induced subgraph. Consequently, $\mathcal{A}_{3,4}$ and $\mathcal{A}_{5,2}$ can not both be nonempty; if it were indeed the case, both $\mathcal{A}_{3,4}$ and $\mathcal{A}_{5,2}$ would also be complete to $\mathcal{A}_{c'}$, and $G$ would then contain an induced claw, namely $x_1,v,w$ and $z$, with $v\in \mathcal{A}_{2,3}$, $w\in \mathcal{A}_{3,4}$ and $z\in \mathcal{A}_{4,5}$. Hence, we may assume, without loss of generality, that $\mathcal{A}_{5,2} = \emptyset$ (the same reasoning applies if we assume that $\mathcal{A}_{3,4} = \emptyset$). But then $\mathcal{A}_{c'}$ and $\mathcal{A}_{3,4}$ must be anti-complete as $G$ would otherwise contain an induced claw (the same as before), and $G$ is consequently $B_1$-EPG (see Fig. \ref{cas1}). 

\begin{center}
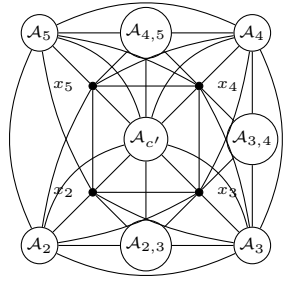
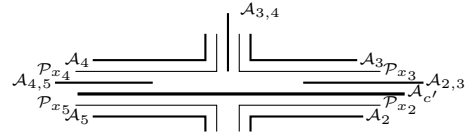
\begin{figure}[h]
\begin{minipage}[b]{0.5\textwidth}
\begin{subfigure}[b]{\linewidth}
\centering
\begin{tikzpicture}[node distance=1cm]
\node[clique] (Ac') {$\mathcal{A}_{c'}$};
\node[circ,label={[label distance=.05cm]180: \tiny $x_2$}] (a2) [below left of=Ac'] {};
\node[circ,label={[label distance=.05cm]0:\tiny $x_3$}] (a3) [below right of=Ac'] {};
\node[circ,label={[label distance=.05cm]0:\tiny $x_4$}] (a4) [above right of=Ac'] {};
\node[circ,label={[label distance=.05cm]180:\tiny $x_5$}] (a5) [above left of=Ac'] {};
\begin{pgfinterruptboundingbox}
\node[clique] (A2) [below left of=a2] {$\mathcal{A}_2$};
\node[clique] (A23) [below right of=a2] {$\mathcal{A}_{2,3}$};
\node[clique] (A3) [below right of=a3] {$\mathcal{A}_3$};
\node[clique] (A34) [above right of=a3] {$\mathcal{A}_{3,4}$};
\node[clique] (A4) [above right of=a4] {$\mathcal{A}_4$};
\node[clique] (A45) [above left of=a4] {$\mathcal{A}_{4,5}$};
\node[clique] (A5) [above left of=a5] {$\mathcal{A}_5$};
\end{pgfinterruptboundingbox}

\draw (Ac') edge[-] (a2)
(Ac') edge[-] (a3)
(Ac') edge[-] (a4)
(Ac') edge[-] (a5)
(Ac') edge[-] (A23)
(Ac') edge[-] (A45)
(Ac') edge[-,bend right=30] (A2)
(Ac') edge[-,bend left=30] (A3)
(Ac') edge[-,bend left=30] (A4)
(Ac') edge[-,bend right=30] (A5)
(a2) edge[-] (a3)
(a2) edge[-] (a5)
(a2) edge[-,bend left=10] (A5)
(a2) edge[-] (A2)
(a2) edge[-] (A23)
(a2) edge[-,bend right=10] (A3)
(a3) edge[-] (a4)
(a3) edge[-,bend left=10] (A2)
(a3) edge[-] (A23)
(a3) edge[-] (A3)
(a3) edge[-] (A34)
(a3) edge[-,bend right=10] (A4)
(a4) edge[-] (a5)
(a4) edge[-,bend left=10] (A3)
(a4) edge[-] (A34)
(a4) edge[-] (A4)
(a4) edge[-] (A45)
(a4) edge[-,bend right=10] (A5)
(a5) edge[-,bend left=10] (A4)
(a5) edge[-] (A45)
(a5) edge[-] (A5)
(a5) edge[-,bend right=10] (A2)
(A2) edge[-,bend left=25] (A5)
(A2) edge[-] (A23)
(A2) edge[-,bend right=25] (A3)
(A3) edge[-] (A23)
(A3) edge[-] (A34)
(A3) edge[-,bend right=25] (A4)
(A4) edge[-] (A34) 
(A4) edge[-] (A45)
(A4) edge[-,bend right=25] (A5)
(A5) edge[-] (A45);
\end{tikzpicture}
\caption{General structure of $G$.}
\end{subfigure}
\end{minipage}
\begin{minipage}[b]{0.5\textwidth}
\begin{subfigure}[b]{\linewidth}
\centering
\begin{tikzpicture}
\coordinate (s) at (0,-.5);
\coordinate (e) at (2,0);
\coordinate (n) at (0,.8);
\coordinate (o) at (-2,0);
\coordinate (c) at (0,0);
\node[invisible] (fake) at (0,-1) {};
\node[invisible] (Ac') at ($(e) + (0.6,0)$) {$\mathcal{A}_{c'}$};
\node[invisible] (A23) at ($(e) + (0.9,0.15)$) {$\mathcal{A}_{2,3}$};
\node[invisible] (a3) at ($(e) + (0.3,0.3)$) {$\mathcal{P}_{x_3}$};
\node[invisible] (A3) at ($(e) + (-.08,0.45)$) {$\mathcal{A}_3$};
\node[invisible] (a2) at ($(e) + (0.3,-0.15)$) {$\mathcal{P}_{x_2}$};
\node[invisible] (A2) at ($(e) + (0,-0.3)$) {$\mathcal{A}_2$};
\node[invisible] (A45) at ($(o) + (-0.6,0.15)$) {$\mathcal{A}_{4,5}$};
\node[invisible] (a4) at ($(o) + (-0.3,0.3)$) {$\mathcal{P}_{x_4}$};
\node[invisible] (A4) at ($(o) + (0,0.45)$) {$\mathcal{A}_4$};
\node[invisible] (a5) at ($(o) + (-0.3,-0.15)$) {$\mathcal{P}_{x_5}$};
\node[invisible] (A5) at ($(o) + (0,-0.3)$) {$\mathcal{A}_5$};
\node[invisible,label=right:\tiny $\mathcal{A}_{3,4}$](A34) at ($(n) + (0,0.3)$) {};

\draw[very thick] (o) -- (Ac'); 
\draw[thick] (A45) -- ($(o) + (1,0.15)$); 
\draw[thick] (A23) -- ($(e) + (-1,0.15)$); 
\draw[thick] (A34) -- ($(n) + (0,-.5)$); 
\draw (a4) -- ($(o) + (1.85,0.3)$) -- ($(n) + (-0.15,0)$); 
\draw (a5) -- ($(o) + (1.85,-0.15)$) -- ($(s) + (-0.15,0)$); 
\draw ($(s) + (0.15,0)$) -- ($(s) + (0.15,.35)$) -- (a2); 
\draw (a3) -- ($(e) + (-1.85,0.3)$) -- ($(n) + (0.15,0)$); 
\draw[thick] (A4) -- ($(o) + (1.7,0.45)$) -- ($(n) + (-0.3,0)$); 
\draw[thick] (A5) -- ($(o) + (1.7,-0.3)$) -- ($(s) + (-0.3,0)$); 
\draw[thick] ($(s) + (0.3,0)$) -- ($(s) + (0.3,.2)$) -- (A2); 
\draw[thick] (A3) -- ($(e) + (-1.7,0.45)$) -- ($(n) + (0.3,0)$); 
\end{tikzpicture}
\caption{A $B_1$-EPG representation of $G$.}
\end{subfigure}
\end{minipage}
\caption{$G$ contains $W_4$ and both $\mathcal{A}_{2,3}$ and $\mathcal{A}_{4,5}$ are nonempty.}
\label{cas1}
\end{figure}
\end{center}

Now, assume without loss of generality, that $\mathcal{A}_{2,3} = \mathcal{A}_{5,2} = \emptyset$ and $\mathcal{A}_{3,4}, \mathcal{A}_{4,5} \neq\emptyset$. We know from the above that $\mathcal{A}_{3,4}$ is anti-complete to $\mathcal{A}_{4,5}$. Also, for all $x \in \mathcal{A}_{c'}$, $x$ must be either complete to $\mathcal{A}_{3,4}$ and anti-complete to $\mathcal{A}_{4,5}$ or, conversely, anti-complete to $\mathcal{A}_{3,4}$ and complete to $\mathcal{A}_{4,5}$, as $G$ would otherwise contain an induced claw; indeed, if there exist $x' \in \mathcal{A}_{3,4}$ and $x'' \in \mathcal{A}_{4,5}$ such that $x$ is nonadjacent (resp. adjacent) to both $x'$ and $x''$, then $x_4,x',x''$ and $x$ (resp. $x,x',x''$ and $x_2$) induce a claw. Thus, we can partition $\mathcal{A}_{c'}$ into subsets $\mathcal{A}_{c'}^i = \{x \in \mathcal{A}_{c'} ~|~ \forall x' \in \mathcal{A}_{i+2,i+3}, xx' \in E\}$ for $i = 1,2$, and, since $G$ is $H_4$-free, both of these subsets are cliques. Assuming $\mathcal{A}_{c'}^1$ and $\mathcal{A}_{c'}^2$ are non empty, there cannot exist $x \in \mathcal{A}_{c'}^1$ and $x' \in \mathcal{A}_{c'}^2$ such that $xx' \not \in E$ since $G$ would otherwise contain an induced $G_2$ (see Fig. \ref{Fig:PCA}); indeed, $x_3, x, x_5, x'$ and $x_2$ would form a 4-wheel, with one vertex of $\mathcal{A}_{3,4}$ adjacent to only $x_3$ and $x$, and one vertex of $\mathcal{A}_{4,5}$ adjacent to only $x_5$ and $x'$. Hence, $\mathcal{A}_{c'}^1 \cup \mathcal{A}_{c'}^2$ is a clique; but then, $G$ contains $H_6$ as an induced subgraph. Thus, we have to assume that exactly one of $\mathcal{A}_{c'}^1$ and $\mathcal{A}_{c'}^2$ is empty, and $G$ is then $B_1$-EPG as an induced subgraph of the previous case. 

Suppose now that only one of the $\mathcal{A}_{j,j+1}$ is non empty, for instance $\mathcal{A}_{3,4}$. If $x \in \mathcal{A}_{c'}$ is adjacent to some vertex $x' \in \mathcal{A}_{3,4}$, then $x$ is complete to $\mathcal{A}_{3,4}$ as $G$ would otherwise contain an induced $G_4$ (see Fig. \ref{Fig:PCA}). We can therefore partition $\mathcal{A}_{c'}$ into two subsets $\mathcal{A}_{c'}^a = \{x \in \mathcal{A}_{c'} ~|~ \forall x' \in \mathcal{A}_{3,4}, xx' \in E\}$ and $\mathcal{A}_{c'}^{na} = \{x \in \mathcal{A}_{c'} ~|~ \forall x' \in \mathcal{A}_{3,4}, xx' \not \in E\}$. Since $G$ does not contain an induced $H_4$, $\mathcal{A}_{c'}^a$ must be a clique; and, since $G$ does not contain an induced claw, $\mathcal{A}_{c'}^{na}$ must also be a clique. If one of $\mathcal{A}_{c'}^a$ and $\mathcal{A}_{c'}^{na}$ is empty, then $G$ is $B_1$-EPG as an induced subgraph of the first case. If both are non empty, then $\mathcal{A}_{c'}^a \cup \mathcal{A}_{c'}^{na}$ is a clique, as $G$ would otherwise contain an induced $G_3$ (see Fig. \ref{Fig:PCA}), and $G$ is consequently $B_1$-EPG (see Fig. \ref{cas2}).

\begin{center}
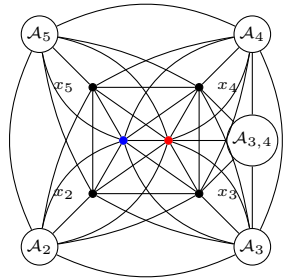
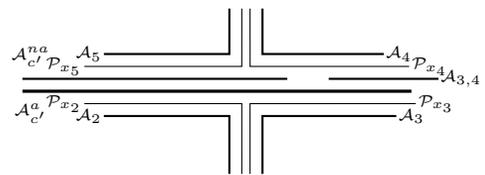
\begin{figure}[h]
\begin{minipage}[b]{0.5\textwidth}
\begin{subfigure}[b]{\linewidth}
\centering
\begin{tikzpicture}[node distance=1cm]
\node[invisible] (c) {};
\node[circ,red] (Ac'a) [right of=c,xshift=-.7cm] {};
\node[circ,blue] (Ac'na) [left of=c,xshift=.7cm] {};
\node[circ,label={[label distance=.05cm]180: \tiny $x_2$}] (a2) [below left of=c] {};
\node[circ,label={[label distance=.05cm]0: \tiny $x_3$}] (a3) [below right of=c] {};
\node[circ,label={[label distance=.05cm]0: \tiny $x_4$}] (a4) [above right of=c] {};
\node[circ,label={[label distance=.05cm]180: \tiny $x_5$}] (a5) [above left of=c] {};
\begin{pgfinterruptboundingbox}
\node[clique] (A2) [below left of=a2] {$\mathcal{A}_2$};
\node[clique] (A3) [below right of=a3] {$\mathcal{A}_3$};
\node[clique] (A34) [above right of=a3] {$\mathcal{A}_{3,4}$};
\node[clique] (A4) [above right of=a4] {$\mathcal{A}_4$};
\node[clique] (A5) [above left of=a5] {$\mathcal{A}_5$};
\end{pgfinterruptboundingbox}

\draw (Ac'a) edge[-] (a2)
(Ac'a) edge[-] (a3)
(Ac'a) edge[-] (a4)
(Ac'a) edge[-] (a5)
(Ac'a) edge[-,bend left=30] (A2)
(Ac'a) edge[-,bend left=30] (A3)
(Ac'a) edge[-,bend right=30] (A4)
(Ac'a) edge[-,bend right=30] (A5)
(Ac'a) edge[-] (Ac'na)
(Ac'a) edge[-] (A34)
(Ac'na) edge[-] (a2)
(Ac'na) edge[-] (a3)
(Ac'na) edge[-] (a4)
(Ac'na) edge[-] (a5)
(Ac'na) edge[-,bend right=30] (A2)
(Ac'na) edge[-,bend right=30] (A3)
(Ac'na) edge[-,bend left=30] (A4)
(Ac'na) edge[-,bend left=30] (A5)
(a2) edge[-] (a3)
(a2) edge[-] (a5)
(a2) edge[-,bend left=10] (A5)
(a2) edge[-] (A2)
(a2) edge[-,bend right=10] (A3)
(a3) edge[-] (a4)
(a3) edge[-,bend left=10] (A2)
(a3) edge[-] (A3)
(a3) edge[-] (A34)
(a3) edge[-,bend right=10] (A4)
(a4) edge[-] (a5)
(a4) edge[-,bend left=10] (A3)
(a4) edge[-] (A34)
(a4) edge[-] (A4)
(a4) edge[-,bend right=10] (A5)
(a5) edge[-,bend left=10] (A4)
(a5) edge[-] (A5)
(a5) edge[-,bend right=10] (A2)
(A2) edge[-,bend left=25] (A5)
(A2) edge[-,bend right=25] (A3)
(A3) edge[-] (A34)
(A3) edge[-,bend right=25] (A4)
(A4) edge[-] (A34) 
(A4) edge[-,bend right=25] (A5);
\end{tikzpicture}
\caption{General structure of $G$ (in red: $\mathcal{A}_{c'}^a$; in blue: $\mathcal{A}_{c'}^{na}$).}
\end{subfigure}
\end{minipage}
\begin{minipage}[b]{0.5\textwidth}
\begin{subfigure}[b]{\linewidth}
\centering
\begin{tikzpicture}[scale=1.1]
\coordinate (s) at (0,-1);
\coordinate (e) at (2,0);
\coordinate (n) at (0,1);
\coordinate (o) at (-2,0);
\coordinate (c) at (0,0);
\node[invisible] (fake) at (0,-2) {};
\node[invisible,label=below:{\tiny $\mathcal{A}_{c'}^a$}] (Ac'a) at ($(o) + (-0.6,0)$) {};
\node[invisible,label=above:{\tiny $\mathcal{A}_{c'}^{na}$}] (Ac'na) at ($(o) + (-0.6,0.15)$) {};
\node[invisible] (A34) at ($(e) + (0.6,0.15)$) {$\mathcal{A}_{3,4}$};
\node[invisible] (a4) at ($(e) + (0.22,0.3)$) {$\mathcal{P}_{x_4}$};
\node[invisible] (A4) at ($(e) + (-.15,0.45)$) {$\mathcal{A}_4$};
\node[invisible] (a3) at ($(e) + (0.3,-0.15)$) {$\mathcal{P}_{x_3}$};
\node[invisible] (A3) at ($(e) + (0,-0.3)$) {$\mathcal{A}_3$};
\node[invisible] (a5) at ($(o) + (-0.2,0.3)$) {$\mathcal{P}_{x_5}$};
\node[invisible] (A5) at ($(o) + (0.1,0.45)$) {$\mathcal{A}_5$};
\node[invisible] (a2) at ($(o) + (-0.2,-0.15)$) {$\mathcal{P}_{x_2}$};
\node[invisible] (A2) at ($(o) + (0.1,-0.3)$) {$\mathcal{A}_2$};

\draw[very thick] ($(Ac'a) + (-0.1,0)$) -- (e); 
\draw[thick] ($(Ac'na) + (-0.1,0)$) -- ($(e) + (-1.5,0.15)$); 
\draw[thick] (A34) -- ($(e) + (-1,0.15)$); 
\draw (a5) -- ($(o) + (1.95,0.3)$) -- ($(n) + (-0.05,0)$); 
\draw (a2) -- ($(o) + (1.95,-0.15)$) -- ($(s) + (-0.05,0)$); 
\draw ($(s) + (0.05,0)$) -- ($(s) + (0.05,0.85)$) -- (a3); 
\draw (a4) -- ($(e) + (-1.95,0.3)$) -- ($(n) + (0.05,0)$); 
\draw[thick] (A5) -- ($(o) + (1.8,0.45)$) -- ($(n) + (-0.2,0)$); 
\draw[thick] (A2) -- ($(o) + (1.8,-0.3)$) -- ($(s) + (-0.2,0)$); 
\draw[thick] ($(s) + (0.2,0)$) -- ($(s) + (0.2,0.7)$) -- (A3); 
\draw[thick] (A4) -- ($(e) + (-1.8,0.45)$) -- ($(n) + (0.2,0)$); 
\end{tikzpicture}
\caption{A $B_1$-EPG representation of $G$.}
\end{subfigure}
\end{minipage}
\caption{$G$ contains $W_4$ and only $\mathcal{A}_{3,4}$ is non empty.}
\label{cas2}
\end{figure}
\end{center}

Finally, if we assume that every $\mathcal{A}_{j,j+1}$ is empty, since $G$ does not contain $H_3$ or a claw as an induced subgraph, we can partition $\mathcal{A}_{c'}$ into two cliques, $\mathcal{A}_{c'}^1$ and $\mathcal{A}_{c'}^2$, which are anti-complete, and again $G$ is $B_1$-EPG (see Fig. \ref{cas3}).
\end{case}

\begin{center}
\begin{figure}[h]
\begin{minipage}[b]{0.5\textwidth}
\begin{subfigure}[b]{\linewidth}
\centering
\begin{tikzpicture}[node distance=1cm]
\node[invisible] (c) {};
\node[circ,red] (Ac'1) [right of=c,xshift=-.7cm] {};
\node[circ,blue] (Ac'2) [left of=c,xshift=.7cm] {};
\node[circ,label={[label distance=.05cm]180: \tiny $x_2$}] (a2) [below left of=c] {};
\node[circ,label={[label distance=.05cm]0: \tiny $x_3$}] (a3) [below right of=c] {};
\node[circ,label={[label distance=.05cm]0: \tiny $x_4$}] (a4) [above right of=c] {};
\node[circ,label={[label distance=.05cm]180: \tiny $x_5$}] (a5) [above left of=c] {};
\begin{pgfinterruptboundingbox}
\node[clique] (A2) [below left of=a2] {$\mathcal{A}_2$};
\node[clique] (A3) [below right of=a3] {$\mathcal{A}_3$};
\node[clique] (A4) [above right of=a4] {$\mathcal{A}_4$};
\node[clique] (A5) [above left of=a5] {$\mathcal{A}_5$};
\end{pgfinterruptboundingbox}

\draw (Ac'1) edge[-] (a2)
(Ac'1) edge[-] (a3)
(Ac'1) edge[-] (a4)
(Ac'1) edge[-] (a5)
(Ac'1) edge[-,bend left=30] (A2)
(Ac'1) edge[-,bend left=30] (A3)
(Ac'1) edge[-,bend right=30] (A4)
(Ac'1) edge[-,bend right=30] (A5)
(Ac'2) edge[-] (a2)
(Ac'2) edge[-] (a3)
(Ac'2) edge[-] (a4)
(Ac'2) edge[-] (a5)
(Ac'2) edge[-,bend right=30] (A2)
(Ac'2) edge[-,bend right=30] (A3)
(Ac'2) edge[-,bend left=30] (A4)
(Ac'2) edge[-,bend left=30] (A5)
(a2) edge[-] (a3)
(a2) edge[-] (a5)
(a2) edge[-,bend left=10] (A5)
(a2) edge[-] (A2)
(a2) edge[-,bend right=10] (A3)
(a3) edge[-] (a4)
(a3) edge[-,bend left=10] (A2)
(a3) edge[-] (A3)
(a3) edge[-,bend right=10] (A4)
(a4) edge[-] (a5)
(a4) edge[-,bend left=10] (A3)
(a4) edge[-] (A4)
(a4) edge[-,bend right=10] (A5)
(a5) edge[-,bend left=10] (A4)
(a5) edge[-] (A5)
(a5) edge[-,bend right=10] (A2)
(A2) edge[-,bend left=25] (A5)
(A2) edge[-,bend right=25] (A3)
(A3) edge[-,bend right=25] (A4)
(A4) edge[-,bend right=25] (A5);
\end{tikzpicture}
\caption{General structure of $G$ (in red: $\mathcal{A}_{c'}^1$; in blue: $\mathcal{A}_{c'}^2$).}
\end{subfigure}
\end{minipage}
\begin{minipage}[b]{0.5\textwidth}
\begin{subfigure}[b]{\linewidth}
\centering
\begin{tikzpicture}
\coordinate (s) at (0,-.5);
\coordinate (e) at (2,0);
\coordinate (n) at (0,.5);
\coordinate (o) at (-2,0);
\coordinate (c) at (0,0);
\node[invisible] (fake) at (0,-1) {};
\node[invisible] (Ac'1) at ($(e) + (0.7,0)$) {$\mathcal{A}_{c'}^1$};
\node[invisible] (a3) at ($(e) + (0.3,0.15)$) {$\mathcal{P}_{x_3}$};
\node[invisible] (A3) at ($(e) + (-.06,0.3)$) {$\mathcal{A}_3$};
\node[invisible] (a2) at ($(e) + (0.3,-0.15)$) {$\mathcal{P}_{x_2}$};
\node[invisible] (A2) at ($(e) + (0,-0.3)$) {$\mathcal{A}_2$};
\node[invisible] (a4) at ($(o) + (-0.3,0.15)$) {$\mathcal{P}_{x_4}$};
\node[invisible] (A4) at ($(o) + (0,0.3)$) {$\mathcal{A}_4$};
\node[invisible] (a5) at ($(o) + (-0.3,-0.15)$) {$\mathcal{P}_{x_5}$};
\node[invisible] (A5) at ($(o) + (0,-0.3)$) {$\mathcal{A}_5$};
\node[invisible,label=right:\tiny $\mathcal{A}_{c'}^2$] (Ac'2) at ($(n) + (0,0.3)$) {};

\draw[very thick] (o) -- (Ac'1); 
\draw[very thick] (s) -- (Ac'2); 
\draw (a4) -- ($(o) + (1.85,0.15)$) -- ($(n) + (-0.15,0)$); 
\draw (a5) -- ($(o) + (1.85,-0.15)$) -- ($(s) + (-0.15,0)$); 
\draw ($(s) + (0.15,0)$) -- ($(s) + (0.15,.35)$) -- (a2); 
\draw (a3) -- ($(e) + (-1.85,0.15)$) -- ($(n) + (0.15,0)$); 
\draw[thick] (A4) -- ($(o) + (1.7,0.3)$) -- ($(n) + (-0.3,0)$); 
\draw[thick] (A5) -- ($(o) + (1.7,-0.3)$) -- ($(s) + (-0.3,0)$); 
\draw[thick] ($(s) + (0.3,0)$) -- ($(s) + (0.3,.2)$) -- (A2); 
\draw[thick] (A3) -- ($(e) + (-1.7,0.3)$) -- ($(n) + (0.3,0)$); 
\end{tikzpicture}
\caption{A $B_1$-EPG representation of $G$.}
\end{subfigure}
\end{minipage}
\caption{$G$ contains $W_4$ and each $\mathcal{A}_{j,j+1}$ is empty.}
\label{cas3}
\end{figure}
\end{center}

\begin{case}[\textit{$G$ contains no induced 4-wheel}] 
Assume henceforth that $\mathcal{A}_{1,2}, \mathcal{A}_{2,3}$ and $\mathcal{A}_{3,1}$ are pairwise anti-complete. If all three subsets are non empty, then for all $x \in \mathcal{A}_c$, there must exist $j \in \{1,2,3\}$ such that $x$ is complete to both $\mathcal{A}_{j-1,j}$ and $\mathcal{A}_{j,j+1}$, and anti-complete to $\mathcal{A}_{j+1,j+2}$ (otherwise $G$ would contain an induced claw). Hence, we can partition $\mathcal{A}_c$ into three subsets $\mathcal{A}_c^j = \{x \in \mathcal{A}_c ~|~ \forall x' \in \mathcal{A}_{j-1,j} \cup \mathcal{A}_{j,j+1}, xx' \in E \text{ and } \forall x' \in \mathcal{A}_{j+1,j+2}, xx' \not \in E\}$ ($j \in \{1,2,3\}$) which must be cliques since $G$ does not contain an induced claw. But then either $\mathcal{A}_c$ is a clique, in which case $G$ is $B_1$-EPG (see Fig. \ref{cas4}), or there exists $x \in \mathcal{A}_c^j$ and $x' \in \mathcal{A}_c^{j+1}$ such that $xx' \not \in E$, and $G$ contains a 4-wheel with $x,x'',x',x_{j+2}$ (for some $x'' \in \mathcal{A}_{j,j+1}$) as its 4-cycle and $x_j$ as its center, which is contrary to our assumption.   

\begin{center}
\begin{figure}[h]
\begin{minipage}[b]{0.5\textwidth}
\begin{subfigure}[b]{\linewidth}
\centering
\begin{tikzpicture}[node distance=.8cm]
\node[clique] (Ac3) {$\mathcal{A}_{c}^3$};
\node[clique] (Ac2) [below right of=Ac3] {$\mathcal{A}_{c}^2$};
\node[clique] (Ac1) [below left of=Ac3] {$\mathcal{A}_{c}^1$};
\node[circ,label=left:{\tiny $x_1$}] (a1) [below left of=Ac1,xshift=-0.5cm] {};
\node[circ,label=right:{\tiny $x_2$}] (a2) [below right of=Ac2,xshift=0.5cm] {};
\node[circ,label=above:{\tiny $x_3$}] (a3) [above of=Ac3] {};
\node[clique] (A12) [below right of=Ac1,yshift=-1cm] {$\mathcal{A}_{1,2}$};
\node[clique] (A23) [above right of=Ac3,xshift=1cm,yshift=0.5cm] {$\mathcal{A}_{2,3}$};
\node[clique] (A31) [above left of=Ac3,xshift=-1cm,yshift=0.5cm] {$\mathcal{A}_{3,1}$};

\draw (Ac1) edge[-] (a1)
(Ac1) edge[-] (a2)
(Ac1) edge[-] (a3)
(Ac1) edge[-] (Ac2)
(Ac1) edge[-] (Ac3)
(Ac1) edge[-] (A31)
(Ac1) edge[-] (A12)
(Ac2) edge[-] (a1)
(Ac2) edge[-] (a2)
(Ac2) edge[-] (a3)
(Ac2) edge[-] (Ac3)
(Ac2) edge[-] (A23)
(Ac2) edge[-] (A12)
(Ac3) edge[-,bend right=15] (a1)
(Ac3) edge[-,bend left=15] (a2)
(Ac3) edge[-] (a3)
(Ac3) edge[-] (A23)
(Ac3) edge[-] (A31)
(a1) edge[-] (a2)
(a1) edge[-] (a3)
(a1) edge[-] (A31)
(a1) edge[-] (A12)
(a2) edge[-] (a3)
(a2) edge[-] (A12)
(a2) edge[-] (A23)
(a3) edge[-] (A23)
(a3) edge[-] (A31);
\end{tikzpicture}
\caption{General structure of $G$.}
\end{subfigure}
\end{minipage}
\begin{minipage}[b]{0.5\textwidth}
\begin{subfigure}[b]{\linewidth}
\centering
\begin{tikzpicture}
\coordinate (e) at (2,0);
\coordinate (n) at (0,2);
\coordinate (o) at (-2,0);
\node[invisible] (fake) at (0,-1.5) {};
\node[invisible] (x1) at ($(e) + (0.6,0)$) {$\mathcal{P}_{x_1}$};
\node[invisible] (Ac1) at ($(o) + (-0.6,-0.15)$) {$\mathcal{A}_c^1$};
\node[invisible] (A12) at ($(e) + (0.9,0.15)$) {$\mathcal{A}_{1,2}$};
\node[invisible] (x2) at ($(e) + (0.35,0.3)$) {$\mathcal{P}_{x_2}$};
\node[invisible] (Ac2) at ($(e) + (-.1,0.44)$) {$\mathcal{A}_c^2$};
\node[invisible] (A31) at ($(o) + (-0.6,0.15)$) {$\mathcal{A}_{3,1}$};
\node[invisible] (x3) at ($(o) + (-0.3,0.3)$) {$\mathcal{P}_{x_3}$};
\node[invisible] (Ac3) at ($(o) + (0,0.45)$) {$\mathcal{A}_c^3$};
\node[invisible,label=right:\tiny $\mathcal{A}_{2,3}$] (A23) at ($(n) + (0,0.3)$) {};

\draw[very thick] (o) -- (x1); 
\draw[thick] (Ac1) -- ($(e) + (0,-0.15)$); 
\draw[thick] (A31) -- ($(o) + (1,0.15)$); 
\draw[thick] (A12) -- ($(e) + (-1,0.15)$); 
\draw[thick] (A23) -- ($(n) + (0,-1)$); 
\draw (x3) -- ($(o) + (1.85,0.3)$) -- ($(n) + (-0.15,0)$); 
\draw (x2) -- ($(e) + (-1.85,0.3)$) -- ($(n) + (0.15,0)$); 
\draw[thick] (Ac3) -- ($(o) + (1.7,0.45)$) -- ($(n) + (-0.3,0)$); 
\draw[thick] (Ac2) -- ($(e) + (-1.7,0.45)$) -- ($(n) + (0.3,0)$); 
\end{tikzpicture}
\caption{A $B_1$-EPG representation of $G$.}
\end{subfigure}
\end{minipage}
\caption{$G$ does not contain $W_4$, all three $\mathcal{A}_{j,j+1}$ are nonempty and $\mathcal{A}_c$ is a clique.}
\label{cas4}
\end{figure}
\end{center}
\end{case}
If we now assume that at least one of the subsets $\mathcal{A}_{j,j+1}$ is empty, then at least one vertex $x_i$ of $C$ is adjacent to every vertex of $G$ and $G$ is consequently an interval graph. Indeed, assume that $G\backslash x_i$ contains an induced cycle $C' = y_1, \cdots, y_l$ with $l>3$. Then, together with $x_i$, it would form an $l$-wheel. But since $G$ is proper, it cannot contain a $k$-wheel with $k > 4$. Hence $l = 4$ and $G$ would contain a 4-wheel which is contrary to our assumption. Therefore, $G$ has no induced cycle of length larger than 3, i.e. $G$ is chordal. Furthermore, if there existed three pairwise nonadjacent vertices in $G\backslash x_i$, then together with $x_i$ they would induce a claw. Hence, $G\backslash x_i$ contains no asteroidal triple, i.e. $G$ is an interval graph (see Theorem \ref{theo:interval}), and therefore $B_1$-EPG.
\end{proof}

From the characterisation of $B_1$-EPR graphs, i.e. intersection graphs of paths on a rectangle of a grid where each path has at most one bend, given in \cite{NCA}, we deduce the following characterisation by a family of minimal forbidden induced subgraphs of PCA graphs which are $B_1$-EPR. It is easily seen that the class of circular arc graphs is exactly the class $B_4$-EPR. The authors of \cite{NCA} further proved that NCA graphs have a bend number, with respect to EPR representations, of 2; hence, since PCA $\subset$ NCA, PCA graphs also have a bend number, with respect to EPR representations, of at most 2. 

\begin{corollary}
Let $G$ be a PCA graph. Then $G$ is $B_1$-EPR if and only if $G$ is $\{W_4, S_3, C_{4k-1}^k, k \geq 2\}$-free.
\end{corollary}

\begin{proof}
It was shown in \cite{NCA} that for $k \geq 2$, $C_{4k-1}^k \not \in B_1$-EPG; a fortiori, $C_{4k-1}^k \not \in B_1$-EPR. We also know from \cite{Ries} that $W_4 \not \in B_1$-EPR and from \cite{Golumbic} that $S_3 \not \in B_1$-EPR.

Conversely, in \cite{PHCA} the authors proved that PCA $\cap$ $\{W_4,S_3\}$-free $=$ PHCA. The result then follows from the fact that PHCA $\cap$ $\{C_{4k-1}^k, k \geq 2\}$-free $\subset$ NHCA $\cap$ $\{C_{4k-1}^k, k \geq 2\}$-free $=$ $B_1$-EPR (see~\cite{NCA}).
\end{proof} 

\section{Conclusion}
\label{sec:conclusion}

In this paper, we present characterisations by (infinite) families of minimal forbidden induced subgraphs for $B_1$-EPG $\cap$ PCA and $B_1$-EPR $\cap$ PCA. This is a first step towards finding a characterisation of the minimal graphs in (CA $\cap ~B_2$-EPG) $\backslash$ (CA $\cap ~B_1$-EPG), a question left open in \cite{NCA}. 

\section*{Acknowledgments}

This research was carried out when Dr M.P. Mazzoleni was visiting the University of Fribourg. The support of this institution is gratefully acknowledged.


\end{document}